\documentclass[conference,10pt,twocolumn]{IEEEtran}
\usepackage{epsfig}
\usepackage{epsf}
\usepackage{underscore}
\usepackage{amsthm}
\usepackage{graphicx}
\graphicspath{ {/} }

\newtheorem{obsrv}{Observation}

\begin{document}

\title{Optimizing {D}ijkstra for real-world performance}

\author{\IEEEauthorblockN{Nimrod Aviram\IEEEauthorrefmark{1},
Yuval Shavitt\IEEEauthorrefmark{2}}
\IEEEauthorblockA{Department of Electrical Engineering,
Tel Aviv University\\
Email: \IEEEauthorrefmark{1}nimrodav@mail.tau.ac.il,
\IEEEauthorrefmark{2}shavitt@eng.tau.ac.il
}}




\maketitle

\begin{abstract}

Using Dijkstra's algorithm to compute the shortest paths in a graph from a single source node to all other nodes is common practice in industry and academia. Although the original description of the algorithm advises using a Fibonacci Heap as its internal queue, it has been noted that in practice, a binary (or $d$-ary) heap implementation is significantly faster. This paper introduces an even faster queue design for the algorithm. 

Our experimental results currently put our prototype implementation at about twice as fast as the Boost implementation of the algorithm on both real-world and generated large graphs. Furthermore, this preliminary implementation was written in only a few weeks, by a single programmer. The fact that such an early prototype compares favorably against Boost, a well-known open source library developed by expert programmers, gives us reason to believe our design for the queue is indeed better suited to the problem at hand, and the favorable time measurements are not a product of any specific implementation technique we employed.

\end{abstract}

\section{Introduction}

Dijkstra's algorithm is a widely-used algorithm for solving the Single Source Shortest Path problem. Namely, given a vertex in a graph with non negative edge weights, compute the distances from this vertex to all other vertices. The algorithm employs a single data structure, a queue.  While any queue implementation will suffice for the correctness of the algorithm, obviously different queue implementations provide different running time complexity, both asymptotically and in practice. Algorithm textbooks mostly recommend using a Fibonacci Heap as the chosen queue implementation because of its fast asymptotic running time \cite{Cormen90}. Practitioners \cite{Boost}, as well as some textbooks \cite{Cormen90}, recommend using a $d$-ary heap as the queue implementation, because of its fast running time in practice.

The Dijkstra algorithm is obviously an important building block in network science where it is used for studying graph characteristics \cite{example-JJ,example-BIU}, large graph clustering \cite{BSWW14}, and many other graph related problems.  It is an important algorithm for the Internet, where it is
part of the shortest path calculation performed by OSPF \cite{ospf} and IS-IS \cite{is-is} protocols. It is also used by many applications in diverse fields, such as image processing \cite{avidan2007seam}, hardware design \cite{example-Cidon}, and many more.

Leveraging the algorithm's properties, we observe that the queue implementation does not have to deal with a general sequence of queries and updates. Indeed, since edge weights are non-negative, once a pop\_min() operation on the queue returned the value $x$, no value smaller than $x$ will ever be inserted (or be present) in the queue. We thus chose a queue implementation based on an array, where all vertices with current distance $d$ are stored in a linked list, whose base is the $d$-th cell of the array (for the clarity of the introduction, consider only the case of integer weights - floating point weights can also be dealt with, as will be explained later). Using this implementation gives O(1) insert() and decrease\_key() time, while the total time for all pop\_min() operations combined is also constant (and in practice, takes a few seconds).

We tested our implementation against Boost, which is one of the most highly regarded and expertly designed C++ library projects in the world \cite{Boost-brag}, and show that it outperforms Boost for both generated Erd\H{o}s-R\'{e}nyi networks, and the real-world mainland USA road network.

Similar performance analysis for several methods of implementing the algorithms, including methods similar to methods we suggest here, was performed by Cherkassky {\em et al.} \cite{cherkassky1996shortest}. However, the experiments in \cite{cherkassky1996shortest} were done some 20 years ago, on the limited hardware available at that time, and consequently, on graphs which could be processed on such hardware - such graphs are considered small nowadays. Just to give a sense of the differences in the order of magnitudes involved, \cite{cherkassky1996shortest} ran the experiments on a SUN Sparc-10 workstation with a 40 MHZ processor and 160 MBytes of memory, whereas most of our experiments were done on a 1600 MHZ machine with 15 Gigabytes of memory. The number of vertices in the graphs we use in our experiments is usually in the millions, whereas \cite{cherkassky1996shortest} have performed only 4 experiments where the number of vertices is above 1 million (and in these experiments, it is only slightly larger than 1 million).
Furthermore, \cite{cherkassky1996shortest} have mainly used a 3-ary heap and a double-bucket queue as queue implementations for the algorithm. A double-bucket queue implementation is similar to the mechanism we term Swap Prevention, described later in this report. In most experiments, they were unable to employ our suggested queue implementation, as it required too much memory. Indeed, with the limited amounts of memory available at the time, this would be expected. However, our results indicate that, assuming enough memory is available, our chosen queue implementation outperforms the Swap Prevention mechanism, and considering it is significantly simpler to implement, it would be the natural choice for the queue implementation. \cite{cherkassky1996shortest} could not run our chosen queue implementation on a single graph where the number of vertices is at least 1 million, and in several instances, they claim it was in fact unemployable even on small graphs (in nowadays standards) where the nodes number in the several thousands (we are surprised by this claim - while we expect memory constraints to be a problem for experiments done at the time, a reasonable implementation of our queue mechanism should be able to run on such graphs using less than the amount of memory they had available). In short, while \cite{cherkassky1996shortest} is obviously an important and beneficial work in this area, it does not analyze the performance of our queue implementation for graphs which would be considered a reasonable benchmark nowadays, on modern hardware.

The rest of the paper is organized as follows: Section \ref{sec:Qimpl} elaborates on the workings of the chosen queue implementation. Section \ref{sec:meas} details our measurements of running time for our implementation and the Boost implementation. Section \ref{sec:extn} discusses a few additional potential improvements for our implementation (including the aforementioned solution for floating point weights). Section \ref{sec:cncld} concludes the paper.

\section{Proposed Queue Implementation}  \label{sec:Qimpl}

We will briefly remind the reader the Dijkstra algorithm for calculating the shortest path between one source vertex and the rest of the vertices \cite{Cormen90}.
The algorithm maintains a queue of vertices, sorted by distance from the starting vertex. The queue is initialized to contain the starting vertex, with its distance of 0 (and formally, all other vertices with distance infinity). In each iteration, a pop\_min() operation is performed on the queue, popping out the vertex with the smallest distance present in the queue. All edges of this vertex are relaxed (in no particular order), while maintaining the distances of target vertices as represented in the queue: i.e., for each edge, if the new distance achieved by adding the edge's weight to the distance of the popped vertex is lower than the previous distance of the target vertex, a decrease\_key() operation is performed on the target vertex as present in the queue (formally distances are decreased from their initial value of infinity to a real value, in practice most implementations first insert vertices to the queue only when their distance is less than infinity). After all edges of the popped vertex are relaxed, the algorithm continues to the next iteration, where it performs another pop\_min() operation and so on.

The algorithm performs (up to) $V$ pop\_min() operations, and (up to) $E$ decrease\_key() operations. Choosing a Fibonacci Heap as the queue implementation gives constant amortized time for decrease\_key(), and O($\log V$) for pop\_min(), bringing the total complexity to O($E+V\log V$). However, as noted by \cite[Ch.\ 21]{Cormen90}, the constant factors hidden behind the Big-O notation for Fibonacci Heap make the running time very long in practice. Choosing a $d$-ary heap for the queue implementation gives O($\log V$) for decrease\_key() and O($\log V$) for pop\_min(), bringing the total complexity to O($E\log V + V\log V$). This is the popular choice for implementations of the algorithm, as exemplified by Boost \cite{Boost}.

Before we describe the queue implementation, we need the following observation:

\begin{obsrv}
\label{non-decreasing}
For Dijkstra's algorithm, once a vertex with distance $x$ returned from the pop\_min() operation, no vertex with distance less than $x$ will ever be present in the queue.
\end{obsrv}

\begin{proof}
Indeed, at the moment the first vertex with distance $x$ returned from pop\_min(), no vertex with distance less than $x$ was present in the queue (otherwise, that other vertex would be the result of pop\_min()).
We now use a proof by induction: suppose at the start of iteration $i$ of the algorithm, no vertex with distance less than $x$ was present in the queue, then at the end of the iteration no such vertex will be present in the queue either (we define "iteration" to mean the sequence of actions taken between two successive pop\_min() operations):
Indeed, name the vertex popped out of the queue in iteration $i$, $v_i$. By our assumption, $v_i$'s distance when popped out, $d$, satisfies $x \le d$. During the iteration, decrease\_key() operations are performed, but the new distances are of the form $d + w$, where $w$ is a weight of some edge, which according to the algorithm's assumptions is non-negative. Hence, the new distances are also at least $x$.
\end{proof}

Obviously, Dijkstra's algorithm performance depends on the queue implementation.  Thus, we now describe our new queue implementation. The queue consists of an array of size MAX\_INT (typically $2^{32}$), where cell $i$ in the array is the anchor of a linked list containing the vertices whose current distance is $i$. Most operations on the queue are somewhat trivial to implement:

\begin{enumerate}

\item init() simply zeroes all cells of the array.

\item insert(vertex $v$, distance $d$) simply inserts $v$ to the linked list in the $d$'th cell of the array.  

Since the list is not ordered $v$ can be placed at the head.

\item decrease\_key(vertex $v$, distance new\_distance) first removes $v$ from the linked list it is currently present in (we use a doubly-linked list for convenience sake), then performs insert($v$, new\_distance).

\end{enumerate}

The only non-trivial operation is pop\_min(), which we now describe:
The queue maintains as one of its internal members a lower bound on the minimal distance of the next pop\_min() operation, which we will call min\_distance\_candidate. This value is initialized to zero, then increased as increasing values are popped out of the queue - indeed, recall that the series of popped distances is (non-strictly) increasing, as proven by Observation~\ref{non-decreasing}. Another way of understanding the role of min\_distance\_candidate is to think of it as pointing to the cell out of which the last popped node was popped.
On pop\_min() start (see Figure~\ref{fig:popmin}), the code scans the cells of the array starting from min\_distance\_candidate, until it reaches a cell containing a non-empty list. The new value of min\_distance\_candidate is the index of this cell. The code then pops the first element from the list, and returns it. 

\begin{figure*}[htb]
\begin{center}
\begin{verbatim}
pop_min(Queue queue)
{
    while (queue.min_distance_candidate <= queue.max_distance_ever_seen) {
        cell = queue.array[min_distance_candidate]
        if not cell.list.empty() {
            result = cell.list.pop_start()
            return result
        }
        queue.min_distance_candidate++
    }
    return NULL
}
\end{verbatim}
\end{center}
\caption{A pseudocode for the pop\_min() operation} \label{fig:popmin}
\end{figure*}

This brings the total running time of the algorithm with our queue implementation to O($E + MAX\_INT$): Recall that Dijkstra's algorithm performs (up to) $V$ pop\_min() operations, and (up to) $E$ decrease\_key() operations. Our running time for decrease\_key() is O(1), contributing the first term of the total running time. As for the pop\_min() operations, we observe that once min\_distance\_candidate points to a non-empty cell in the array, pop\_min() is O(1). The total amount of operations, over the entire running time of the algorithm, advancing min\_distance\_candidate from its initial value of 0 to its final value of MAX\_INT is clearly MAX\_INT.

Furthermore, min\_distance\_candidate does not need to reach its maximum theoretical value of MAX\_INT - once the queue is empty, pop\_min() can return NULL and the algorithm is done. This can be implemented by maintaining the current number of vertices present in the queue, then returning NULL once that number reaches 0. Thus, the final value of min\_distance\_candidate is the distance of the farthest vertex from the starting vertex, which we will designate by $U$. This brings the total complexity to O($E + U$).
In our current implementation, the queue instead maintains the largest distance that was ever present in the queue, termed max\_distance\_ever\_seen, then terminates when min\_distance\_candidate surpasses max\_distance\_ever\_seen. Maintaining this value is important for avoiding unnecessary initializations of array cells. We emphasize that the actual memory interactions with the array only occur with cells up to max\_distance\_ever\_seen, such that memory regions above max\_distance\_ever\_seen aren't even physically allocated.
Adding another value of the number of vertices present to the implementation is quite easy, but currently unimplemented, as we suspect it won't further reduce our running time. 

Even in cases where the possible value of U is close to $2^{32}$ (this means using integers for storing distances is only marginally sufficient and thus may not be an appropriate choice), we emphasize at this early section of the paper that going over an array of size $2^{32}$ may sound prohibitively expensive, but in practice isn't: on our development machine (a strong machine for personal use), it takes about 50 Seconds to do so, while some memory needs to be swapped to the hard-drive. We elaborate on the matter later, but emphasize even now that this doesn't make our approach prohibitively expensive.

\section{Measurements} \label{sec:meas}

We now present measurements of the running time of our implementation, compared to the running time of the Boost implementation. We stress that achieving comparable performance to Boost is quite a feat in and of itself, since we're comparing code that was developed in a few weeks to a highly regarded, well polished library. Our code is publicly available at \cite{github}, and we encourage further experimentation with it. All time measurements were done on an Intel Core-i7 machine with 16 gigabytes of RAM (except for the protein network graph, see below). For Boost time measurements, we tested their implementation on each graph with 4 different heap implementations they recommend, then took the shortest time as the "Boost time". 

We benchmarked the implementations against a set of graphs generated from the Erd\H{o}s-R\'{e}nyi Model \cite{Erdos}. Figure \ref{fig:ER-time} shows that our implementation outperforms Boost on all tested graph sizes and densities. The runtime speedup ranges from 1.47 to 8.

\begin{figure}[ht!]
\begin{center}
\includegraphics[width=0.45\textwidth,natwidth=1201,natheight=900]{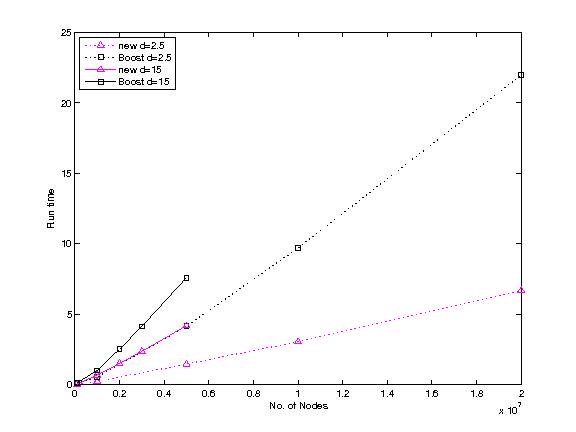}
\caption{Run time comparison between our implementation and Boost for generated Erd\H{o}s-R\'{e}nyi graphs. }
\label{fig:ER-time}
\end{center}
\end{figure}

Furthermore, we benchmarked both implementations on generated Barab\'{a}si-–Albert graphs \cite{barabasi1999emergence}, with the parameter $m$, the number of new edges per new vertex, ranging between 2 and 10, and with 10 million vertices. The weights are uniformly selected between 1 and 1000. Figure \ref{fig:BA-time} shows that in practice, there is little difference in the running time for different values of $m$; our implementation typically runs in about one milli-Second, and Boost typically runs in one and a half hundredth of a second, giving a speedup of about 15. Figure \ref{fig:BA-size} compares both implementations' running time for $m = 2$, while the number of vertices, $n$, grows. Our implementation's running time is always lower than Boost's, but it is not necessarily increasing with $n$. Such low running times, of typically less than 25 milli-Second, seem to indicate the graphs don't start to exercise the asymptotic behavior of our implementation. Boost's running time, on the other hand, is increasing with $n$, which seems to indicate that as $n$ will increase further, the running time will also increase further. For the graphs described above, each mark point represents the average of, at least, 20 random experiments.

\begin{figure}[ht!]
\begin{center}
\includegraphics[width=0.45\textwidth,natwidth=1201,natheight=900]{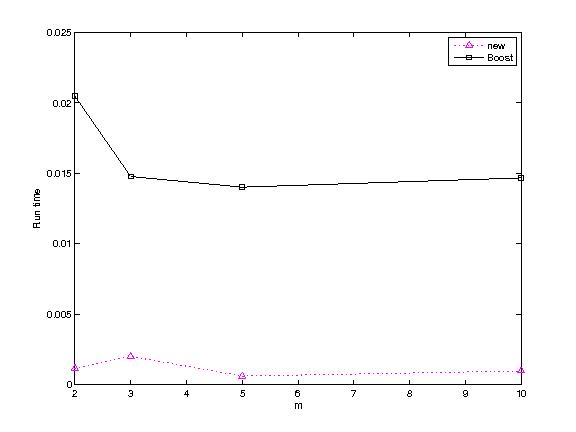}
\caption{Run time comparison between our implementation and Boost for generated Barab\'{a}si-–Albert graphs. }
\label{fig:BA-time}
\end{center}
\end{figure}

\begin{figure}[ht!]
\begin{center}
\includegraphics[width=0.45\textwidth,natwidth=1201,natheight=900]{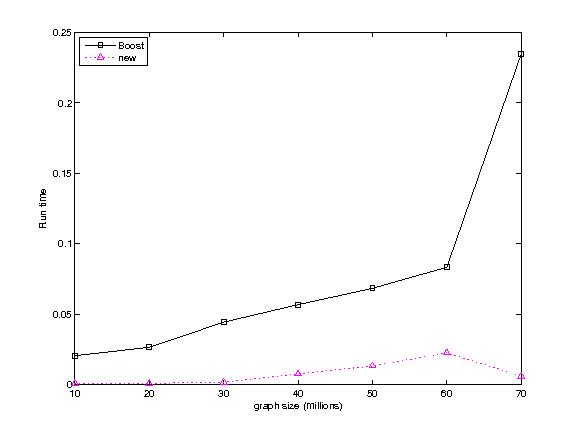}
\caption{Run time comparison between our implementation and Boost for generated Barab\'{a}si-–Albert graphs, for $m=2$. }
\label{fig:BA-size}
\end{center}
\end{figure}

We also benchmarked the implementations against the graph of the entire mainland USA road network, obtained from \cite{challenge9}. Our implementation typically runs in 2-3 Seconds, while the Boost implementation (again, taken as the shortest time between 4 possible heap implementations - in our experience, the variance of runtime between the different heap implementations is quite small) typically runs in 6-7 Seconds. Figure \ref{fig:usa} shows the performance over 1000 randomly selected starting vertices in that graph (note that the X axis designates a randomly chosen starting vertex for each point, not the 1000 vertices with the smallest indices) - clearly, our implementation runs faster than Boost on this graph regardless of the starting vertex.

Additionally, we attempted to benchmark our implementation against Boost on the protein network \cite{string-paper}, which can be found at \cite{string}, that translates to a graph with about five million nodes and 664 million edges. 
Such a large graph strains the memory requirements on our machine. Nevertheless, our implementation's runtime ranges from 0.0019 Seconds to 0.082 Seconds on a sample of 13 starting vertices.  
The Boost implementation initially hanged the machine due to large memory requirements for storing the graph and an initial data structure, which caused constant swapping and required quite an effort to help it executing.  We managed to get it to run for one case in 0.4 Sec for an instance that took our implementation 0.02 Seconds to run, a factor of 20 speedup for our implementation.  

\begin{figure}[ht!]
\begin{center}
\includegraphics[width=0.45\textwidth,natwidth=1201,natheight=900]{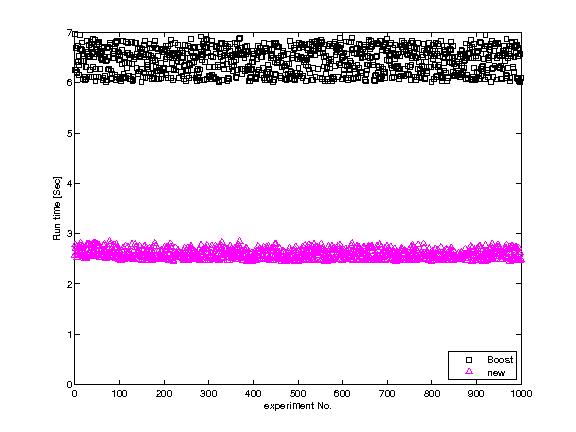}
\caption{Run time comparison between our implementation and Boost for the full USA road network. }
\label{fig:usa}
\end{center}
\end{figure}

\begin{table}
\begin{tabular}{| r | c | l | l |l|}
\hline
Vertices & Density & Our time & Boost time & speedup \\ 
    &   &   [Sec.] & [Sec.] & \\ \hline\hline
100,000 & 2.5 & 0.01 & ~0.08 & 8.00 \\ \hline
1,000,000 & 2.5 & 0.23 & ~0.53 & 2.30 \\ \hline
5,000,000 & 2.5 & 1.45 & ~4.14 & 2.86 \\ \hline
10,000,000 & 2.5 & 3.05 & ~9.7 & 3.18 \\ \hline
20,000,000 & 2.5 & 6.68 & 21.97 & 3.29 \\ \hline\hline
100,000 & 15 & 0.03 & ~0.11 & 3.67 \\ \hline
1,000,000 & 15 & 0.68 & ~1 & 1.47 \\ \hline
2,000,000 & 15 & 1.5 & ~2.51 & 1.67 \\ \hline
3,000,000 & 15 & 2.36 & ~4.12 & 1.75 \\ \hline
5,000,000 & 15 & 4.18 & ~7.59 & 1.82\\ \hline\hline
road USA & & &  & \\
23,947,347 & 2.44 & 2.57 & 6.25 & 2.43\\ \hline   
\end{tabular}
\caption{Run time and speedup comparison between our algorithm and Boost for generated E-R graphs, and USA road network. }
\label{tab:time}
\end{table}

\section{Improvements And Extensions}  \label{sec:extn}
We will now discuss two potential extensions to our implementation: dealing with floating point weights, and swap-prevention.  To the best of our knowledge, this is the first time the techniques presented in the previous sections are suggested for dealing with real numbers.

{\bf Dealing with floating point weights}: Naturally, as presented so far, our implementation lacks the capability to deal with floating point weights, as well as with weights larger than $2^{32}$. For the latter case, we suggest simply using a 32-bit floating point as the chosen weight representation, and accepting the resulting minimal loss of precision. We acknowledge that loss of precision might be unacceptable for some use cases of the algorithm, but postulate that this is quite rare. The same applies for 64-bit floating point representation: our method requires accepting the loss of precision resulting from switching to 32-bit floating point representation.

Focusing on 32-bit floating point weights, we observe that in essence, the queue depends on the weights being integers solely for the purpose of iterating over the weights in a monotonically increasing order. This can also be achieved for floating point values: a (positive) floating point value is in essence an ordered pair of a mantissa and an exponent, and comparing two values is simply comparing them lexicographically, exponent first and mantissa second. Thus, the floating point value corresponding to exponent $e$ and mantissa $m$, is larger than exactly $m + e \cdot M$ other floating point values, where $M$ is the number of possible mantissas. Therefore, the queue implementation can simply change to have cell $i$ contain a linked list of all vertices whose distance is the $i$-th smallest floating point value. It is easy to show that this preserves the correctness of the queue.

It should further be noted that most use cases of the algorithm can likely use a floating point representation that is 24 bits or less. For example, using 10 bits for the mantissa allows for a precision of 3 decimal digits past the decimal point, which probably suffices for the vast majority of cases, and using even just 6 bits for the exponent allows for orders of magnitude between $2^{-32}$ and $2^{32}$. We postulate that only rare use cases cannot accept a similar level of precision (this example fits in 16 bits. Using 24 bits and allocating the remaining 8 bits according to the use case's needs allows for an even greater level of precision). Obviously, assuming such level of precision is appropriate, our implementation's running time will be O($E + 2^{24}$), which in practice equals O($E$) since the O($E$) pointer manipulation operation is more meaningful in time consumption than reading the 64 mega-byte ($2^{24}$) long array.

{\bf Swap-Prevention}: as briefly mentioned earlier, going over the full possible range of distances, 0 to $2^{32}$, is in fact not prohibitively expensive even on a (not extraordinarily strong) 16 gigabyte RAM machine, where some memory required for the queue representation will have to be swapped to the hard drive. Just to give a sense of proportions, the total amount of memory required for our implementation in this case is 19.3 gigabytes (the major memory requirements are 16 gigabytes for the queue and 3.2 gigabytes statically allocated for the graph), and about 3 gigabytes are required for the operating system and other software running on the machine, which means roughly about 6 gigabytes will need to be swapped to the hard drive in the course of a single run of the program;  this requires about 50 Seconds. Swapping obviously can be avoided entirely by purchasing more RAM, which is quite inexpensive nowadays, and bringing the total to 32 gigabytes.

While accepting the cost of swapping or purchasing more RAM may be acceptable for some cases, the problem can also be avoided by employing a mechanism we term Swap-Prevention: We can limit the required memory for the queue to small values (almost arbitrarily small ones - even smaller than the CPU cache size). The way Swap-Prevention works is by dividing the array to equally sized pieces which we term "chunks". We wish to guide the reader by example, making the mechanism much clearer. Suppose each chunk is "16 bits", or $2^{16}$ long. Thus, the array is divided to $2^{16}$ chunks. The first chunk contains vertices with distances 0 to $2^{16} - 1$, the second one holds vertices with distances $2^{16}$ to $2^{17} - 1$, etc.
At any given moment, a single chunk is "active", meaning its $2^{16}$ cells occupy an array accordingly-sized in memory. The active chunk is the chunk containing the min\_distance\_candidate cell. Non-active chunks are "condensed", each of those chunks occupies a single linked list, containing all of the vertices present in the chunk. 

\begin{enumerate}

\item Inserting (and similarly, decrease\_key) is quite easy: if the new distance of a vertex falls inside the active chunk, the vertex is inserted to the linked list of vertices with that exact distance. Otherwise, the vertex is inserted to the single linked list containing all the vertices of its new, non-active chunk (this indeed causes the temporary "inconvenience" of having vertices with different distances in the same linked list - an inconvenience which will be dealt with later).

\item pop\_min() is somewhat more complex: if min\_distance\_candidate points to a non-empty cell, obviously a vertex can be popped out from the linked list and returned. Otherwise, min\_distance\_candidate is advanced in the hope of finding a non-empty cell inside the current active chunk. If min\_distance\_candidate reached the end of the chunk, the next (non-empty) chunk needs to be "expanded" into the array: The queue goes over the vertices of the "condensed" single linked list, and inserts each vertex to the appropriate cell inside the array. The chunk then becomes active, min\_distance\_candidate is reset to point to the first cell of the array, and the simpler pop\_min() case can now be executed.

\end{enumerate}

Clearly, the memory requirements when employing Swap-Prevention are an array whose size is the chunk size, and an additional array containing an anchor for each chunk. Thus, the memory requirements are $CHUNK\_SIZE + NUM\_OF\_CHUNKS$. Note that $CHUNK\_SIZE \cdot NUM\_OF\_CHUNKS = MAX\_INT$ must hold. Optimizing for memory consumption obviously gives the optimal chunk size as the square root of MAX\_INT, typically $2^{16}$. Note that there is no need to choose this particular value, any value for the chunk size will work (as long as there aren't divisibility considerations).

Swap-Prevention was implemented by us \cite{Swap-Prevention-Commit}, and somewhat surprisingly, found to actually impede performance by a factor of about 2. Our original aim was to fit the queue implementation inside the CPU cache, which is perfectly possible (our CPU, an Intel Core-i7, has 8 megabytes of cache, and choosing a $2^{16}$ chunk size keeps us under a single megabyte), but this neglects the fact that the graph representation is typically several gigabytes of memory, so cache misses will be abundant regardless. We welcome additional experimentation with the Swap-Prevention code. We also note that in cases where a 16-bit floating point (or even 16-bit integer) representation suffices, the queue can be made extremely small using Swap-Prevention, able to fit even in extremely tight memory requirements typically found in embedded systems. 

\section{Conclusion} \label{sec:cncld}
We presented a novel queue implementation well-suited for the Dijkstra's algorithm. Using this implementation, the algorithm's runtime is O($E + U$), where $U$ is the distance of the vertex farthest from the starting vertex. A prototype implementation compares favorably to Boost, a well-known, widely-used library. We released the code to make it available to the research community.


\section*{Acknowledgements}
This report was supported in part by a scholarship from the Ministry of Science and Technology of Israel.

\bibliographystyle{acm}
\bibliography{dijk}

\end{document}